\documentclass[envcountsame,runningheads,orivec]{llncs}
\usepackage{hyperref}
\usepackage{amsmath,amssymb}
\usepackage[T1]{fontenc}
\usepackage[utf8]{inputenc}
\usepackage[color=yellow]{todonotes}
\usepackage{enumerate}
\usepackage{MnSymbol}
\usepackage[margin=1.6in]{geometry}

\usepackage{tikz}
\usetikzlibrary{calc}
\usetikzlibrary{patterns}

\newcommand{\conv}{\operatorname{conv}}
\newcommand{\xc}{\operatorname{xc}}

\newcommand{\R}{\mathbb{R}}
\newcommand{\Z}{\mathbb{Z}}

\newcommand{\ocp}{\operatorname{ocp}}
\newcommand{\oct}{\operatorname{oct}}
\newcommand{\STAB}{\operatorname{STAB}}

\newcommand{\LP}{\operatorname{LP}}

\newcommand{\stab}{\STAB}
\newcommand{\zerovec}{\mathbf{0}}
\newcommand{\onevec}{\mathbf{1}}

\newcommand{\sub}{P}
\newcommand{\slack}{Q}

\renewcommand{\ge}{\geqslant}
\renewcommand{\le}{\leqslant}

\makeatletter
\renewcommand{\@Opargbegintheorem}[4]{%
  #4\trivlist\item[\hskip\labelsep{#3#2\@thmcounterend}]}
\makeatother

\begin{document}
\title{Extended Formulations for Stable Set Polytopes of Graphs Without Two Disjoint Odd Cycles}

\author{Michele Conforti \inst{1} \and Samuel Fiorini \inst{2} \and Tony Huynh \inst{3} \and Stefan Weltge \inst{4}}

\authorrunning{M.~Conforti et al.}
\titlerunning{Extended Formulations for Stable Set Polytopes of OCP-1 Graphs}

\institute{
    Dipartimento di Matematica, Universit\`a degli Studi di Padova, Padova, Italy\\
    \email{conforti@math.unipd.it}
    \and
    D\'epartement de Math\'ematique, Universit\'e libre de Bruxelles, Brussels, Belgium\\
    \email{sfiorini@ulb.ac.be}
    \and
    School of Mathematics, Monash University, Melbourne, Australia\\
    \email{tony.bourbaki@gmail.com}
    \and
    Fakult\"at f\"ur Mathematik, Technische Universit\"at M\"unchen, Munich, Germany\\
    \email{weltge@tum.de}
}

\maketitle

\begin{abstract}
Let $G$ be an $n$-node graph without two disjoint odd cycles. The algorithm of Artmann, Weismantel and Zenklusen (STOC'17) for bimodular integer programs can be used to find a maximum weight stable set in $G$ in strongly polynomial time. Building on structural results characterizing sufficiently connected graphs without two disjoint odd cycles, we construct a size-$O(n^2)$ extended formulation for the stable set polytope of $G$.
\end{abstract}


\section{Introduction}
It is a classic result that integer programs with a totally unimodular constraint matrix $A$ are solvable in strongly polynomial time.  
Very recently, Artmann, Weismantel and Zenklusen~\cite{AWZ17} generalized this to \emph{bimodular} matrices $A$. These include all matrices with all subdeterminants in $\{-2,-1,0,1,2\}$.  As noted in~\cite{AWZ17}, this has consequences for the maximum weight stable set problem in graphs
as follows.  

Let $\stab(G)$ be the \emph{stable set polytope} of a graph $ G $ 
and note that
\[
    \stab(G) = \conv \{ x \in \{0,1\}^{V(G)} \mid Mx \le \mathbf{1} \},
\]
where $M \in \{0,1\}^{E(G) \times V(G)}$ is the edge-node incidence matrix of $ G $. It is well-known that the maximum absolute value of a subdeterminant of $M$ is equal to $2^{\ocp(G)}$, where $\ocp(G)$ is the maximum number of (node-)disjoint odd cycles of $G$ (see \cite{GKS95}). Therefore, the bimodular algorithm of~\cite{AWZ17} can be used to efficiently compute a maximum weight stable set in a graph without two disjoint odd cycles.  

Although the bimodular algorithm is extremely powerful, it provides limited insight on which properties of graphs with $ \ocp(G) \le 1 $ are relevant to derive efficient algorithms for graphs with \emph{higher} odd cycle packing number. Indeed, in light of recent work linking the complexity and structural properties of integer programs to the magnitude of its subdeterminants~\cite{Tardos86,DF94,VC09,BDEHN14,EV17,AWZ17,PSW}, it is tempting to believe that integer programs with bounded subdeterminants can be solved in polynomial time. This would imply in particular that the stable set problem on graphs with $\ocp(G) \leqslant k$ is polynomial for every fixed $k$. Conforti, Fiorini, Huynh, Joret, and Weltge~\cite{CFHJW19} recently proved this is true under the additional assumption that $G$ has bounded (Euler) genus.\footnote{The \emph{Euler genus} of graph $G$ is the minimum of $|E(G)|-|V(G)|-|F(G)|-2$, taken over all embeddings of $G$ in a (orientable or non-orientable) surface, where $F(G)$ denotes the set of faces of $G$ with respect to the embedding.}

Furthermore, by itself the bimodular algorithm does not imply any linear description of the stable set polytope of graphs $G$ with $\ocp(G) = 1$. It turns out that for such graphs, $\STAB(G)$ may have many facets with high coefficients that do not seem to allow a ``nice'' combinatorial description in the original space. While stable set polytopes have been studied for several classes of graphs, very little is known about $\STAB(G)$ when $\ocp(G) = 1$. 

Our main result is to show that every such stable set polytope admits a compact description in an ``extended'' space.
An \emph{extended formulation} of a polyhedron $ P $ is a description of the form $ P = \{ x \mid \exists y : Ax + By \le b \} $ whose \emph{size} is the number of inequalities in $ Ax + By \le b $. The \emph{extension complexity} of $P$, denoted $\xc (P)$, is the minimum size of an extended formulation of $P$. Our main result is the following.

\begin{theorem} \label{thm:main}
For every $n$-node graph $G$ with $\ocp(G) \leqslant 1$, $\STAB (G)$ admits a size-$O(n^2)$ extended formulation. Moreover, this extended formulation can be constructed in polynomial time. 
\end{theorem}

Note that this does not follow from the main result of~\cite{AWZ17}. As noted in~\cite[Thm.~5.4]{CevallosWeltgeZenklusen}, integer hulls of bimodular integer programs can have exponential extension complexity. Moreover, Theorem~\ref{thm:main} also does not follow from~\cite{CFHJW19} since here we are dealing with \emph{arbitrary} graphs $G$ with $\ocp(G) \leqslant 1$.

Our proof is based on a characterization of graphs with $\ocp(G) \le 1$ due to Lov\'asz (see Seymour~\cite{seymour95}).  Kawarabaya\-shi and Ozeki~\cite{KO13} later gave a short, purely graph-theoretical proof of the same result.  Before stating Lov\'asz' theorem, we need a few more definitions.  

The \emph{odd cycle transversal number} of a graph $G$, denoted $\oct (G)$, is the minimum size of a set of nodes $X$ such that $G - X$ is bipartite.  The \emph{projective plane} is the surface obtained from a closed disk by identifying antipodal points on its boundary. An embedding of a graph $G$ in a surface is an \emph{even-face embedding} if every face of $G$ is an open disk bounded by an even cycle of $G$. We point out that graphs that are even-face embedded in a surface are in particular $2$-connected, since the embedding yields an ear-decomposition.

\begin{theorem}[Lov\'asz, cited in~\cite{seymour95}] \label{thm:Lovasz}
Let $G$ be a $4$-connected graph with $ \ocp(G) \le 1 $. Then
\begin{enumerate}[(i)]
    \item $\oct (G) \le 3$, or
    \item $G$ has an even-face embedding in the projective plane.  
\end{enumerate}
\end{theorem}

Note that if a graph $G$ satisfies $(i)$ of Theorem~\ref{thm:Lovasz}, then $\STAB (G)$ has a linear-size extended formulation since it is the convex hull of the union of at most eight polytopes described by nonnegativity and edge constraints.
If $G$ is $4$-connected and satisfies~$(ii)$, we will show that $\STAB(G)$ admits a quadratic-size extended formulation.
To this end, we will consider an affine mapping of $\STAB(G)$ into the edge space $\R^{E(G)}$.
This novel view allows us to identify points in $\R^{V(G)}$ with circulations in a certain orientation of the dual graph of $G$, which can be compactly described using extended formulations.
This approach has been also taken in~\cite{CFHJW19}, where related embeddings of graphs on surfaces with higher genus were considered.
As even-face embeddings refer to a slightly different notion and since the projective plane allows for much more direct arguments, we provide a self-contained proof here.
This yields the proof of Theorem~\ref{thm:main} for the case of $4$-connected graphs.

For non-$4$-connected graphs with $\ocp(G) \le 1$ we have to address their decomposition in a non-trivial manner.
We will deal with polyhedral aspects of performing $2$- and $3$-sums, by exploiting special properties of such graphs.
We remark that for arbitrary graphs, performing multiple $ k $-sums does not preserve small extended formulations for the respective stable set polytopes, even for $k = 2$.  Note that our definition of $k$-sums allows some edges of the clique to be deleted (which is necessary for the proofs).  If no edges of the clique are deleted during a $k$-sum, then it is well-known that small extended formulations are preserved by Chv\'atal's clique cutset lemma~\cite[Theorem 4.1]{Chvatal}.

Our polyhedral analysis also crucially relies on new insights about $\stab(G)$ that hold for all graphs $G$.  For example, the structure of facets of stable set polytopes (see Lemma~\ref{lem:edge_induced_weights}), and the transformation of $\stab(G)$ into the edge space.
We believe that this perspective can be equally beneficial for future investigations of (general) stable set polytopes.

Each step of our proof can be executed in polynomial time, and therefore the extended formulation can be constructed in polynomial time.
Moreover, our proof can also be turned into a direct, purely graph-theoretic strongly polynomial time algorithm for the stable set problem in graphs $G$ with $\ocp(G) \le 1$.

\paragraph{Outline} In Section~\ref{secStructure}, we build on Theorem~\ref{thm:Lovasz} and its signed version due to Slilaty~\cite{slilaty07} to describe the structure of graphs without two disjoint odd cycles. Roughly, we prove that each  such graph $G$ either has $\oct(G) \leq 3$ or can be obtained from a graph $H_0$ having an even-face embedding in the projective plane by gluing internally disjoint bipartite graphs $T_1$, \ldots, $T_\ell$ ``around'' $H_0$ in a certain way. Section~\ref{secProjective_planar_case} contains a proof of Theorem~\ref{thm:main} for the case of graphs with an even-face embedding in the projective plane. The general case is treated in Sections~\ref{secGeneral_case} and \ref{secProofEF} by a delicate argument using certain gadgets $H_1$, \ldots, $H_\ell$ ``simulating'' the bipartite graphs $T_1$, \ldots, $T_\ell$. 


\section{The structure of graphs without two disjoint odd cycles}
\label{secStructure}

In this section we show that every graph without two disjoint odd cycles either has a small odd cycle transversal or has a structure that we will exploit later. For this purpose we use the notion of separations. A \emph{$ k $-separation} of a graph $ G $ is an ordered pair $ (G_0, G_1) $ of edge-disjoint subgraphs of $ G $ with $ G = G_0 \cup G_1 $, $ |V(G_0) \cap V(G_1)| = k $, and $E(G_0), E(G_1), V(G_1) \setminus V(G_0), V(G_0) \setminus V(G_1)$ all non-empty.  We say that a $k$-separation is \emph{linked} if for all distinct nodes $u, v$ of $V(G_0) \cap V(G_1)$ there exists a $u$--$v$ path in $G_1$ whose internal nodes are disjoint from $G_0$.

\begin{definition} \label{defnComb}
A \emph{star structure} of a graph $ G $ is a set of subgraphs $ H_0,T_1,\dots,T_\ell $ of $G$  such that for all $i \in [\ell]$: $ T_i$ is bipartite, $ (H_0 \cup_{j \ne i} T_j, T_i) $ is a linked $ k $-separation of $ G $ with $ k \le 3 $, and $V(T_i) \cap V(T_j) \subseteq V(H_0)$ for all $j \ne i$.
\end{definition}


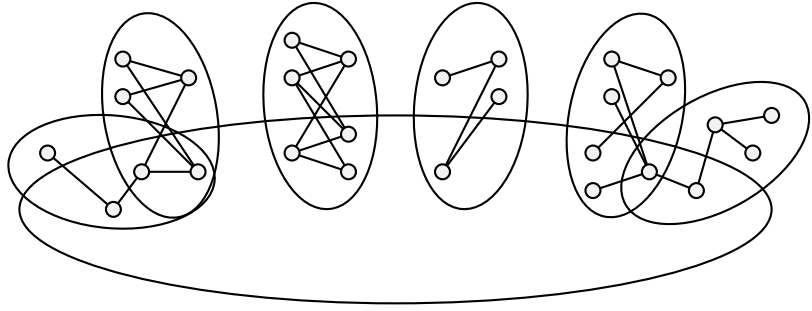
\begin{figure}[ht]
\centering
\begin{tikzpicture}[inner sep=2pt,scale=.5]
\tikzstyle{vtx}=[circle,draw,thick,fill=black!5]
\draw[thick] (9.5,2) ellipse (10 and 2.5);
\draw[thick,shift={(1.95,3)},rotate=85] (0,0) ellipse (1.5 and 2.75);
\draw[thick] (2,2) -- (0.25,3.5);
\draw[thick] (2,2) -- (2.75,3);
\node[vtx] at (2,2) {};
\node[vtx] at (0.25,3.5) {};
\draw[thick,shift={(3.25,4.5)},rotate=10] (0,0) ellipse (1.5 and 2.75);
\draw[thick] (2.75,3) -- (4.25,3);
\draw[thick] (2.75,3) -- (4,5.5);
\draw[thick] (2.25,5) -- (4.25,3);
\draw[thick] (2.25,5) -- (4,5.5);
\draw[thick] (2.25,6) -- (4.25,3);
\draw[thick] (2.25,6) -- (4,5.5);
\node[vtx] at (2.75,3) {};
\node[vtx] at (2.25,5) {};
\node[vtx] at (2.25,6) {};
\node[vtx] at (4.25,3) {};
\node[vtx] at (4,5.5) {};
\draw[thick,shift={(7.5,4.75)},rotate=5] (0,0) ellipse (1.5 and 2.75);
\draw[thick] (6.75,3.5) -- (8.25,3);
\draw[thick] (6.75,3.5) -- (8.25,4);
\draw[thick] (6.75,3.5) -- (8.25,6);
\draw[thick] (6.75,5.5) -- (8.25,3);
\draw[thick] (6.75,5.5) -- (8.25,4);
\draw[thick] (6.75,5.5) -- (8.25,6);
\draw[thick] (6.75,6.5) -- (8.25,4);
\draw[thick] (6.75,6.5) -- (8.25,6);
\node[vtx] at (6.75,3.5) {};
\node[vtx] at (6.75,5.5) {};
\node[vtx] at (6.75,6.5) {};
\node[vtx] at (8.25,3) {};
\node[vtx] at (8.25,4) {};
\node[vtx] at (8.25,6) {};
\draw[thick,shift={(11.5,4.75)},rotate=-5] (0,0) ellipse (1.5 and 2.75);
\draw[thick] (10.75,3) -- (12.25,5);
\draw[thick] (10.75,3) -- (12.25,6);
\draw[thick] (10.75,5.5) -- (12.25,6);
\node[vtx] at (10.75,3) {};
\node[vtx] at (10.75,5.5) {};
\node[vtx] at (12.25,5) {};
\node[vtx] at (12.25,6) {};
\draw[thick,shift={(15.625,4.5)},rotate=-12] (0,0) ellipse (1.5 and 2.75);
\draw[thick] (14.75,2.5) -- (16.25,3);
\draw[thick] (14.75,3.5) -- (16.75,5.5);
\draw[thick] (15.25,5) -- (16.25,3);
\draw[thick] (15.25,6) -- (16.25,3);
\draw[thick] (15.25,6) -- (16.75,5.5);
\draw[thick] (16.25,3) -- (17.5,2.5);
\node[vtx] at (14.75,2.5) {};
\node[vtx] at (14.75,3.5) {};
\node[vtx] at (15.25,5) {};
\node[vtx] at (15.25,6) {};
\node[vtx] at (16.25,3) {};
\node[vtx] at (16.75,5.5) {};
%
\draw[thick,shift={(18,3.5)},rotate=-60] (0,0) ellipse (1.5 and 2.75);
\draw[thick] (17.5,2.5) -- (18,4.25);
\draw[thick] (18,4.25) -- (19.5,4.5);
\draw[thick] (18,4.25) -- (19,3.5);
\node[vtx] at (17.5,2.5) {};
\node[vtx] at (19.5,4.5) {};
\node[vtx] at (19,3.5) {};
\node[vtx] at (18,4.25) {};
\end{tikzpicture}
\caption{A star structure.} \label{fig:comb_structure}
\end{figure}

For our structural result we will also use the notion of signed graphs. A \emph{signed graph} is a pair $(G, \Sigma)$ where $G$ is a graph and $\Sigma \subseteq E(G)$. A subgraph of $ G $ is said to be \emph{$ \Sigma $-odd} if it contains an odd number of edges in $ \Sigma $, and is \emph{$\Sigma$-even} otherwise. The \emph{odd cycle packing number} of a signed graph $(G, \Sigma) $ is the maximum number of disjoint $ \Sigma $-odd cycles in $(G, \Sigma) $, and is denoted by $\ocp(G,\Sigma)$. A signed graph $(G,\Sigma)$ is \emph{balanced} if $\ocp (G, \Sigma)=0$. The \emph{odd cycle transversal number} of $ (G, \Sigma) $ is the minimum number of nodes in $ (G, \Sigma) $ intersecting every $ \Sigma $-odd cycle in $ (G, \Sigma) $, and is denoted by $ \oct(G, \Sigma) $. An embedding of a signed graph $(G, \Sigma)$ in a surface is an \emph{even-face embedding} if every face of $(G, \Sigma)$ is an open disk bounded by a $\Sigma$-even cycle of $(G, \Sigma)$. Graphs in this section may have parallel edges. 

In the definition below, $\biguplus$ is used to denote the edge-disjoint union of graphs.

\begin{definition} \label{defnHStar}
Let $ G $ be a graph with star structure $ H_0, T_1,\dots,T_\ell $.  For each $i \in [\ell]$, let $S_i=V(H_0) \cap V(T_i)$ and note that there is a signed clique $ (K_i, \Sigma_i) $ with $ V(K_i) = S_i $ such that $ (K_i \biguplus T_i, \Sigma_i \biguplus E(T_i)) $ is balanced. The signed graph $ (H^+, \Sigma) $ is then defined via $ H^+ := H_0 \biguplus K_1 \biguplus \dots \biguplus K_\ell $ and $ \Sigma := E(H_0) \biguplus \Sigma_1 \biguplus \dots \biguplus \Sigma_\ell $.
\end{definition}

The structural result is the following.

\begin{theorem}
    \label{thmStructure}
    Let $ G $ be a graph with $ \ocp(G) = 1 $ and $ \oct(G) \ge 4 $. Then $ G $ admits a star structure $ H_0, T_1,\dots,T_\ell $,  such that $S_1, \dots, S_\ell$ and $(H^+, \Sigma) $ from Definition~\ref{defnHStar} have the following properties:
    \begin{itemize}
        \item $ S_i$ is not a subset of $S_j$ for all distinct $i,j \in [\ell]$,
        \item $ (H^+, \Sigma) $ has an even-face embedding in the projective plane, and
        \item the nodes of each $ S_i $ are on the boundary of some face of the embedding.
    \end{itemize}
\end{theorem}

In order to obtain the above statement, we will use a finer version of Theorem~\ref{thm:Lovasz} that is suited for signed graphs, due to Slilaty~\cite{slilaty07}. The latter result was previously known by Gerards, Lov\'asz, and others, but~\cite{slilaty07} is the first time it appears in print.   

\begin{theorem}[Slilaty~\cite{slilaty07}] \label{thm:Slilaty}
Let $ (G, \Sigma) $ be a $ 4 $-connected signed graph with $ \ocp(G, \Sigma) \le 1 $. Then
\begin{enumerate}[(1)]
    \item $ \oct(G, \Sigma) \le 3 $ or
    \item $ (G, \Sigma) $ has an even-face embedding in the projective plane.
\end{enumerate}
\end{theorem}

\begin{lemma} \label{lemCoreOCPOCT}
    Let $ G $ be a graph with star structure $ H_0,T_1,\dots,T_\ell $ and let $ (H^+, \Sigma) $ be as in Definition~\ref{defnHStar}. Then $ \ocp(H^+, \Sigma) = \ocp(G) $ and $ \oct(H^+, \Sigma) \ge \oct(G) $.
\end{lemma}
\begin{proof}
    Let $ (K_1,\Sigma_1),\dots,(K_\ell, \Sigma_\ell) $ be as in Definition~\ref{defnHStar}.   We first show $ \ocp(H^+, \Sigma) \le \ocp(G)$. Let $ C_1,\dots,C_k $ be disjoint $ \Sigma $-odd cycles in $ (H^+, \Sigma) $. For each $ j \in [k] $ we will construct an odd cycle $ C_j' $ in $ G $ by performing the following modifications to $C_j$ for each $i \in [\ell]$.
    First observe that $ C_j $ contains at most two edges from $ K_i $. Otherwise, $C_j=K_i$ since $|V(K_i)| \le 3 $, which contradicts that $ (K_i, \Sigma_i) $ is balanced.   If $ C_j $ uses only one edge $ uv $ of $ K_i $, we replace $uv$ by a $ u $--$ v $ path $ P_i $ in $ T_i $ whose internal nodes are disjoint from $V(K_i)$.  Note that $P_i$ exists since the separation $(H_0 \cup_{h \ne i} T_h, T_i)$ is linked.  If $ C_j $ uses two edges from $ K_i $, say $ uv $ and $ vw $, we replace $\{uv, vw\}$ by a $ u $--$ w $ path $ P_i $ in $ T_i $. Again, $P_i$ exists by linkedness.  If $C_j$ is edge-disjoint from $K_i$, we make no modification to $C_j$ for $i$ and set $P_i=\emptyset$. Since $(K_i \uplus T_i, \Sigma_i \uplus E(T_i))$ is balanced, $|E(P_i)|$ and $|E(C_j ) \cap E(K_i) \cap \Sigma|$ have the same parity.
    Therefore, $|E(C_j')|$ and $|E(C_j) \cap \Sigma|$ have the same parity, and so $C_j'$ is odd.   Finally, the cycles $ C_1',\dots,C_k' $ are still disjoint since for each $ i \in [\ell] $ there is at most one cycle $ C_j $ that contains an edge from $ K_i $ (due to $ |V(K_i)| \le 3)$.  Thus, $ \ocp(H^+, \Sigma) \le \ocp(G) $.
    
    For the other inequalities, consider an arbitrary odd cycle $C'$ in $G$.  By reversing the construction from the previous paragraph, there exists a $\Sigma$-odd cycle $C$ in $ (H^+, \Sigma) $ with $V(C) \subseteq V(C')$.   It follows that $ \ocp(H^+, \Sigma) \ge \ocp(G) $ and $ \oct(H^+, \Sigma) \ge \oct(G) $.  
     \qed
\end{proof}

\begin{proof}[Proof of Theorem~\ref{thmStructure}]
Let $ G $ be a graph with $ \ocp(G) = 1 $ and $ \oct(G) \ge 4 $. Let $ H_0,T_1,\dots,T_\ell $ be a star structure with $( |V(H_0)|, \ell) $ lexicographically minimal.  Note that such a star structure exists since $ G $ is a star structure of itself.

 Suppose there exist distinct $ i,j \in [\ell] $ such that $ S_j \subseteq S_i  $. Since $ |S_i| \le 3 $ and $ \oct(G) \ge 4 $, $ G - S_i $ contains an odd cycle $ C $. Note that $ C $ is not a subgraph of $ T_i \cup T_j $ because $ T_i $ and $ T_j $ are both bipartite and hence every odd cycle of $ T_i \cup T_j $ must intersect $ S_i $. Since $ \ocp(G) \le 1 $ this implies that $ T_i \cup T_j $ is bipartite, a contradiction to the minimality of $ \ell $.
 
 Suppose $ (H^+, \Sigma) $ is not $ 4 $-connected. Let $ ((H_1, \Upsilon_1), (H_2, \Upsilon_2)) $ be a separation of $ (H^+, \Sigma) $ with $X:= V(H_1) \cap V(H_2)$ and $|X| \le 3$. By Lemma~\ref{lemCoreOCPOCT}, $ \ocp(H^+, \Sigma) = 1 $ and $\oct(H^+, \Sigma) \ge 4 $. Therefore, exactly one of $(H_1, \Upsilon_1)-X$ or $(H_2, \Upsilon_2)-X$ is balanced.  By symmetry, we may assume that $(H_2, \Upsilon_2)-X$ is balanced, and by taking $|V(H_2)|$ to be minimal we may assume that $ ((H_1, \Upsilon_1), (H_2, \Upsilon_2)) $ is linked. Recall that $ (H^+, \Sigma) $ arises from $ H_0 $ by adding (balanced) signed cliques $ (K_1,\Sigma_1),\dots,(K_\ell,\Sigma_\ell) $ corresponding to the bipartite graphs $ T_1,\dots,T_\ell$. Replacing each $ (K_i, \Sigma_i) $ by the bipartite graph $ T_i $, we see that $ G $ admits a star structure $ H_0', T_1',\dots,T_q' $ where $ V(H_0') = V(H_1) $ , a contradiction to the minimality of $| V(H_0)| $.
Since $ \ocp(H^+, \Sigma) = 1 $ and $ \oct(H^+, \Sigma) \ge 4 $, Theorem~\ref{thm:Slilaty} implies that $ (H^+, \Sigma) $ has an even-face embedding in the projective plane.

Suppose $ S_i $ is not contained on the boundary of a face of the embedding for some $i \in [\ell]$. Since all nodes in $ S_i $ are adjacent in $ (H^+, \Sigma) $, this implies $ |S_i| = 3 $. 
But now, $ S_i $ is a cutset of $ (H^+, \Sigma) $, contradicting that $ (H^+, \Sigma) $ is $ 4 $-connected. \qed
\end{proof}


\section{The projective planar case}
\label{secProjective_planar_case}

In this section, we give a compact extended formulation for $\stab(G)$ when $G$ has an even-face embedding in the projective plane.  The results in this section follow from~\cite{CFHJW19}, where graphs embedded in bounded genus surfaces are considered.  However, to keep our exposition self-contained, we include all proofs.  Moreover, since the projective planar case is devoid of many of the technical difficulties for the bounded genus case, we hope that this section can serve as a gentler introduction to these techniques than~\cite{CFHJW19}.

Our starting point is the unbounded polyhedron 
\[
    \sub(G) := \conv \{ x \in \Z^{V(G)} \mid Mx \le \onevec \},
\]
where $M$ is the edge-node incidence matrix of $G$. Its relationship to $\STAB(G)$ is as follows.

\begin{lemma}
\label{lemStabVsSub}
For every graph $ G $, $ \stab(G) = \sub(G) \cap [0,1]^{V(G)} $.
\end{lemma}

Thus, given an extended formulation for $ \sub(G) $ we only need to add at most $ 2|V(G)| $ linear inequalities (describing $ [0,1]^{V(G)}) $ to obtain one for $ \stab(G) $.  To prove Lemma~\ref{lemStabVsSub}, we exploit the following result.

\begin{lemma}
    \label{lem:P(G)_projection}
    Let $G$ be any graph, and let $v_0$ be any fixed node of $G$. The projection of $\sub(G)$ onto the coordinates indexed by $V(G-v_0)$ equals $\sub(G-v_0)$.
\end{lemma}

\begin{proof}
    To see that the projection of $\sub(G)$ is contained in $\sub(G-v_0)$, it suffices to prove that every integer point $x \in \sub(G)$ projects to a point in $\sub(G-v_0)$. Let $x' \in \Z^{V(G-v_0)}$ be the projection of $x$. Then, for every edge $vw$ in $G-v_0$ we have $x'_v + x'_w = x_v + x_w \leqslant 1$ and hence $x' \in \sub(G-v_0)$, as claimed.

    Conversely, let $x' \in \Z^{V(G-v_0)}$ be any integer point in $\sub(G-v_0)$. Consider a point $x \in \Z^{V(G)}$ that projects to $x'$. By decreasing $x_{v_0}$ by a sufficiently large integer amount, we may assume that $x_v + x_w \leqslant 1$ for all edges $vw \in E(G)$. Hence, $x$ is an integer point in $\sub(G)$. We conclude that the projection of $\sub(G)$ contains $\sub(G-v_0)$. \qed
\end{proof}

\begin{proof}[Proof of Lemma~\ref{lemStabVsSub}]
    It suffices to show that the polytope $\sub(G) \cap [0,1]^{V(G)}$ is integer. We establish this claim by induction on the number of nodes of $G$. The statement is clearly true if $G$ consists of a single node. Now assume that $G$ has at least two nodes, and the statement holds for all proper induced subgraphs of $G$. We have to show that it holds for $G$ itself.

    We may assume that $G$ is connected. If not, then let $G_1$ and $G_2$ be disjoint and proper induced subgraphs of $G$ whose union is equal to $G$, and in particular $\sub(G) = \sub(G_1) \times \sub(G_2)$. By the induction hypothesis we know that $ \sub(G_1) \cap [0,1]^{V(G_1)} $ and $ \sub(G_2) \cap [0,1]^{V(G_2)} $ are integer and hence $\sub(G) \cap [0,1]^{V(G)} = (\sub(G_1) \cap [0,1]^{V(G_1)}) \times (\sub(G_2) \cap [0,1]^{V(G_2)}) $ is integer as well.

    Now consider any vertex $x^*$ of $ \sub(G) \cap [0,1]^{V(G)} $. Let $V_0 \subseteq V(G) $ denote the set of nodes $v$ such that $x^*_v = 0$ and $V_1 \subseteq V(G)$ denote the set of nodes $v$ such that $x^*_v = 1$.

    Let us first consider the case that $V_0 = \emptyset$. We claim that also $V_1 = \emptyset$. Suppose not, so $x^*_v = 1$ for some $v \in V(G)$. Let $w \in V(G)$ be a neighbor of $v$. Such a node exists since $G$ is connected and has at least two nodes. Since $x^*_w \geqslant 0 $ and $ x^*_v + x^*_w \leqslant 1 $, we obtain $x^*_w = 0$, a contradiction to $V_0 = \emptyset$. So, in this case we would have $0 < x^*_v < 1$ for all $v \in V(G)$, implying that $x^*$ is a vertex of $\sub(G)$. However, vertices of $\sub(G)$ are integer and hence we arrive at another contradiction.

    Thus, there must exist a node $v_0 \in V_0$. By Lemma~\ref{lem:\sub(G)_projection}, the projection of $x^*$ onto the coordinates indexed by $V(G-v_0)$ belongs to $\sub(G-v_0) \cap [0,1]^{V(G-v_0)}$. By induction, this projection can be expressed as a convex combination of 0/1-points in $\sub(G-v_0)$. Thus, there exist stable sets $S_1,\dotsc,S_k$ of $G-v_0$ and coefficients $\lambda_1, \dotsc, \lambda_k \in \R_{\ge 0}$ such that $\sum_{i} \lambda_i = 1$ and
    \[
    x^*_v = \sum_{i : v \in S_i} \lambda_i
    \]
    for all $v \in V(G-v_0)$. Since $x^*_{v_0} = 0$, the equation above also holds for $v = v_0$. Now, every stable set of $G-v_0$ is also a stable set of $G$. It follows that $x^*$ is a convex combination of $0/1$-points in $\sub(G)$. \qed
\end{proof}

Thus, it suffices to study $\sub(G)$ instead of $\stab(G)$. To this end, it is convenient to switch from the node space of $G$ to the edge space of $G$ by considering the affine map $ \sigma : \R^{V(G)} \to \R^{E(G)} $ defined via
\[
    \sigma(x) := \onevec - Mx\,.
\]
Under $\sigma$, a vector $ x \in \R^{V(G)} $ is mapped to $ y = \sigma(x) \in \R^{E(G)} $ where $ y_{vw} = 1 - x_v - x_w $ for every edge $ vw \in E(G) $. Since $\sigma$ is invertible if and only if $G$ has no bipartite component, we can focus on $\slack(G) := \sigma(P(G))$. 

We provide an extended formulation for $\slack(G)$, assuming that $G$ is even-face embedded in the projective plane. Let $G^*$ be the dual graph of $G$. An orientation $D$ of the edges of $G^*$ is called \emph{alternating} if in the local cyclic ordering of the edges incident to each dual node $ f $, the edges alternatively leave and enter $ f $. If $G$ admits an alternating orientation of its dual graph, we will relate the points in $\slack(G)$ to certain circulations in $D$, which, as we will see, gives rise to a compact extended formulation. 


\begin{lemma}
\label{lemQofGWithCirculations}
Let $G$ be a non-bipartite graph that is even-face embedded in the projective plane. Then the dual graph $G^*$ of $G$ has an alternating orientation.  
\end{lemma}

The proof of Lemma~\ref{lemQofGWithCirculations} relies on the fact that the parity of every cycle in $G$ is determined by a certain topological property of the cycle.  Before going into more details, we need the notion of a signature of an embedded graph.  

Let $G$ be a graph embedded in the projective plane.  
Each $u \in V(G)$ has a neighborhood that is a disk $\Delta_u$. By arbitrarily choosing one of the two orientations of each $\Delta_u$, we obtain a local orientation at each $u \in V(G)$. Now, take any edge $e=vw$, and let $\Delta_e$ be a disk containing $e$. The local orientations at $v$ and $w$ are either consistent or inconsistent within $\Delta_e$. We define the \emph{signature} $\Sigma \subseteq E(G)$ as the set of edges $e = vw$ such that the local orientations at nodes $v$ and $w$ are inconsistent. Note that the signature depends on the choice of local orientations.  However, it turns out that all signatures are `equivalent' in a sense which we now describe.  

A cycle of $G$ is said to be \emph{$1$-sided} if it is $\Sigma$-odd and \emph{$2$-sided} otherwise. 
Notice that changing the local orientations at some nodes corresponds to resigning on a cut. Therefore, the property of being $1$-sided or $2$-sided does not depend on the local orientations, and only on the embedding of $G$. We point out that a cycle is $2$-sided if and only if it bounds a disk. This follows from the fact that in the projective plane, $2$-sided cycles are always surface separating. A proof of this fact and other basic properties of curves in the projective plane can be found in~\cite{MoharThom}.

\begin{lemma} \label{2sided}
Let $G$ be a non-bipartite graph that is even-face embedded in the projective plane.  Then a cycle of $G$ is $2$-sided if and only if it is even.
\end{lemma}

\begin{proof}
A cycle of $G$ is a \emph{facial cycle} if it is the boundary of a face of $G$. Let $C$ be a $2$-sided cycle of $G$. Then $C$ bounds a closed disk $\Delta$ in the projective plane.  Observe that $E(C)$ is the symmetric difference of all $E(F)$, where $F$ ranges over all facial cycles of $G$ contained in $\Delta$.  Since $G$ is even-face embedded, $|E(F)|$ is even for all such $F$, and hence $|E(C)|$ is also even.  

For the other direction, let $C$ be a $1$-sided cycle of $G$.  It is well-known that every $1$-sided cycle of $G$ is the symmetric difference of $C$ together with some facial cycles of $G$.  Therefore, if $C$ is even, then every $1$-sided cycle of $G$ is even. Since we have already established that every $2$-sided cycle of $G$ is even, $G$ is bipartite, which is a contradiction. Therefore, $C$ is odd.   
\qed
\end{proof}




\begin{proof}[Proof of Lemma~\ref{lemQofGWithCirculations}]
Let $T$ be a spanning tree of $G$, and let $\Delta$ be a disk containing $T$. Pick local orientations at each node in order to put all the edges of $T$ in the corresponding signature $\Sigma$. Seen in $\Delta$, this corresponds to picking a proper $2$-coloring of $T$, assigning to the nodes in one color class the clockwise orientation and to the nodes in the other color class the counterclockwise orientation.

Now take any edge $e$ of $G$ that is not an edge of $T$. Let $C$ denote the unique cycle in $T + e$. By Lemma~\ref{2sided}, $C$ is $\Sigma$-even if and only if $C$ is even.  Since $f \in \Sigma$ for all $f \in C \setminus \{e\}$, if follows that $e \in \Sigma$. Therefore, for this choice of local orientations, we have $\Sigma=E(G)$. We will use these local orientations to define an orientation of $G^*$ as follows. 



 Let $F$ be a face of $G$ and $v_F$ be the corresponding dual node in $G^*$.  Let $uv$ be an edge of $G$ on the boundary of $F$ and $uv^*$ be the corresponding dual edge in $G^*$. Let $\Delta_u$ and $\Delta_v$ be neighbourhoods of $u$ and $v$ such that $uv$ intersects the boundaries of $\Delta_u$ and $\Delta_v$ exactly once, say at $u'$ and $v'$, respectively.  Let $\vec{O_u}$ and $\vec{O_v}$ be the orientations of the boundaries of $\Delta_u$ and $\Delta_v$ given by the local orientations chosen for $u$ and $v$.  Since $uv \in \Sigma$, it follows that at $u'$ and $v'$, $\vec{O_u}$ and $\vec{O_v}$ are either both entering $F$ or both leaving $F$. If they are both entering $F$ we orient $uv^*$ towards $v_F$, and if they are both leaving $F$ we orient $uv^*$ away from $v_F$. Since every face of $G$ is even, note that $v_F$ is an even-degree vertex of $G^*$, and by construction, the orientation is alternating at $v_F$.

Moreover, since every edge of $G$ is in the signature $\Sigma$, repeating the same construction for each face of $G$ gives a well-defined alternating orientation of $G^*$.   \qed
\end{proof}

Let $G$ be even-face embedded in the projective plane and $D$ be an alternating orientation of $G^*$. Note that there is a bijection between the edges of $G$ and the arcs of $D$.  Therefore, we may regard a vector $y \in \R^{E(G)}$ as a vector in $\R^{A(D)}$, and vice versa.  With this identification, $\slack(G)$ turns out to be the convex hull of all non-negative integer circulations of $D$ that satisfy one additional constraint.  

 \begin{lemma} \label{lem:odd_circulations}
Let $G$ be a non-bipartite graph that is even-face embedded in the projective plane, $D$ be an alternating orientation of $G^*$, and $C$ be an arbitrary odd cycle in $G$.  Then  
\[
\slack(G) = \conv \{ y \in \Z^{E(G)}_{\ge 0} \mid  y \text{ is a circulation in } D 
    \text{ and } y(E(C)) \text{ is odd} \}.
    \]
\end{lemma}
\begin{proof}
    Setting $L := \{ y \in \R^{E(G)} \mid y \text{ is a circulation in } D \}$, we have to show that
    \begin{equation*}
        \slack(G) = \conv \{ y \in L \cap \Z^{E(G)}_{\ge 0} \mid y(E(C)) \text{ is odd} \} =: \slack'(G)
    \end{equation*}
    holds.
    To this end, we first show that $\sigma(\R^{V(G)}) = L$ holds.
    To see that $\sigma(\R^{V(G)}) \subseteq L$ let $x \in \R^{V(G)}$ and consider $y = \sigma(x) \in \R^{E(G)}$.
    Let $f \in V(D)$ be any node of the dual graph.
    As $G$ is even-face embedded, $f$ is bounded by an even cycle $C_f$ in $G$.
    Let $e_1 = v_0 v_1,\, e_2 = v_1 v_2, \dots,\, e_{2k} = v_{2k-1} v_{2k}$ denote the edges of $C_f$, where $v_0 = v_{2k}$.
    Note that the edges of $C_f$ correspond to the arcs of $D$ that are incident to $f$.
    As $D$ is alternating, we have
    \begin{align*}
        \pm \left( y(\delta^{\mathrm{in}}(f)) - y(\delta^{\mathrm{out}}(f)) \right)
        = \sum_{i=1}^{2k} (-1)^i y_{e_i}
        = \sum_{i=1}^{2k} (-1)^i (1 - x_{v_{i-1}} - x_{v_i}) = 0.
    \end{align*}
    Thus, $y$ is a circulation in $D$ and we obtain $\sigma(\R^{V(G)}) \subseteq L$.
    To see that we indeed have $\sigma(\R^{V(G)}) = L$, notice that $\sigma(\R^{V(G)})$ and $L$ are linear subspaces, and hence it suffices to show that their dimensions coincide.
    To this end, we make use of Euler's formula for the projective plane, which yields
    \[
        |V(G)| = |E(G)| - |V(D)| + 1.
    \]
    Moreover, we need the basic fact that $\dim(L) = |E(D)| - |V(D)| + 1$.
    This implies
    \[
        \dim(L) = |E(G)| - |V(D)| + 1 = |V(G)| = \dim(\R^{V(G)}),
    \]
    as claimed.

    We next show that $\slack(G) \subseteq \slack'(G)$ holds.
    To this end, it suffices to show that for every $x \in \Z^{V(G)}$ with $Mx \le \onevec$ the vector $y = \sigma(x) = \onevec - Mx $ is contained in $\slack'(G)$.
    Clearly, we have $y \in \sigma(\R^{V(G)}) = L$ as well as $y \in \Z^{E(G)}_{\ge 0}$.
    It remains to show that $y(E(C))$ is odd.
    Let $e_1 = v_0 v_1,\, e_2 = v_1 v_2, \dots,\, e_{2k+1} = v_{2k} v_{2k+1}$ denote the edges of $C$, where $v_0 = v_{2k+1}$.
    We see that
    \[
        \sum_{i=1}^{2k} (-1)^i y_{e_i} = \sum_{i=1}^{2k} (-1)^i (1 - x_{v_{i-1}} - x_{v_i}) = 2x_{v_0} - 1
    \]
    is odd, and so is $\sum_{i=1}^{2k} y_{e_i} = y(E(C))$.

    It remains to show that $\slack'(G) \subseteq \slack(G)$ holds.
    To this end, it suffices to show that every $y \in L \cap \Z^{E(G)}_{\ge 0}$ with $y(E(C))$ odd is contained in $\slack(G)$.
    As $y \in L = \sigma(\R^{V(G)})$ there is some $x \in \R^{V(G)}$ with $\onevec - Mx = \sigma(x) = y$.
    The nonnegativity of $y$ implies $Mx \le \onevec$.
    It remains to show that $x$ is integer.
    As $y$ is integer and $y_{vw} = 1 - x_v - x_w$ for every edge $vw \in E(G)$, we see that $x_v$ is integer if $x_w$ is integer for any neighbor $w$ of $v$.
    Since $G$ is connected it thus suffices to show that $x_v$ is integer for a single $v \in V(G)$.
    This holds true since $2x_{v_0} - 1 = \sum_{i=1}^{2k} (-1)^i y_{e_i}$ is an odd integer, and hence $x_{v_0}$ is integer. \qed
\end{proof}

In view of Lemma~\ref{lem:odd_circulations} it remains to consider the following final lemma.

\begin{lemma} \label{lem:EF_odd_circulations} 
    Let $D$ be a directed graph with node set $V$ and arc set $A$, and let $X \subseteq A$.
    Then the convex hull of nonnegative integer circulations $y$ in $D$ with $y(X)$ odd admits an extended formulation of size $O(|V||A|)$.
\end{lemma}
\begin{proof}
    Let $P$ denote the convex hull of nonnegative integer circulations $y$ in $D$ with $y(X)$ odd.
    Consider the auxiliary graph $D'$ with node set $V' := V \times \{0,1\}$ and arcs from $(v,p)$ to $(w,p)$ for every $(v,w) \in A \setminus X$, $p \in \{0,1\}$ as well as arcs from $(v,p)$ to $(w,1-p)$ for every $(v,w) \in A \cap X$, $p \in \{0,1\}$.
    For each $v \in V$ let $Q_v$ denote the polyhedron of (uncapacitated) unit flows from $(v,0)$ to $(v,1)$ in $D'$.
    Moreover, let $Q$ denote the convex hull of the union of all $Q_v$ for $v \in V$.
    Recall that each $Q_v$ can be described using $|A'| = 2|A|$ linear inequalities (plus some linear equations) and hence, by applying Balas' theorem~\cite{Balas79}, we obtain an extended formulation for $Q$ of size $O(|V||A|)$.

    It remains to show that $P$ is a linear image of $Q$.
    To this end, consider the map $\pi : \R^{A'} \to \R^A$ defined via
    \begin{align*}
        \pi(z)_{(v,w)} & := z_{((v,0),(w,0))} + z_{((v,1),(w,1))} \quad \text{for every } (v,w) \in A \setminus X, \text{ and}\\
        \pi(z)_{(v,w)} & := z_{((v,0),(w,1))} + z_{((v,1),(w,0))} \quad \text{for every } (v,w) \in A \cap X.
    \end{align*}
    For each $v \in V$ consider $P_v := \pi(Q_v)$.
    The recession cones of $P$ and each $P_v$ are equal to the set of all nonnegative circulations in $D$, and hence the recession cones of $P$ and $\pi(Q)$ coincide.
    Thus, it suffices to show that every vertex of $\pi(Q)$ is contained in $P$ and every vertex of $P$ is contained in $\pi(Q)$.

    Let $y$ be a vertex of $\pi(Q)$.
    Then it is the image of a vertex $z$ of $P_v$ for some $v$.
    In particular, $z$ is an \emph{integer} unit flow from $(v,0)$ to $(v,1)$ in $D'$.
    It is now easy to check that $y$ is a nonnegative integer circulations in $D$ with $y(X)$ odd.

    Conversely, let $y$ be a vertex of $P$.
    Thus, it is a nonnegative integer circulation in $D$ with $y(X)$ odd.
    Moreover, it is the characteristic vector of a directed cycle in $D$.
    Indeed, decompose $y = y_1 + \dots + y_k$ into characteristic vectors of directed cycles in $D$.
    As $y(X)$ is odd, we may assume that $y_1(X)$ is odd.
    In particular $y_1 \in P$.
    Note that $y_2 + \dots + y_k$ is in the recession cone of $P$ and hence, since $y$ is a vertex of $P$ we must have $y_2 + \dots + y_k = \zerovec$.

    Suppose that node $v$ is contained in the cycle corresponding to $y$.
    Then it is easy to see that $y$ is the image (under $\pi$) of the characteristic vector of a path from $(v,0)$ to $(v,1)$ in $D'$, and hence $y \in \pi(Q_v) \subseteq \pi(Q)$. \qed
\end{proof}

We are ready to prove the final result of this section. 


\begin{theorem} \label{thm:EFcore}
    Let $G$ be an $n$-node graph that is even-face embedded in the projective plane.  Then $\STAB (G)$ has a size-$O(n^2)$ extended formulation.
\end{theorem}

\begin{proof}
We may assume that $G$ is non-bipartite since the statement is trivial otherwise. By
Lemma \ref{lemQofGWithCirculations}, $G^*$ has an alternating orientation. As before, we denote this orientation by $D$. 

Next, we use Lemmas~\ref{lem:odd_circulations} and \ref{lem:EF_odd_circulations} to obtain an extended formulation for $\slack(G)$. The size of this formulation is $O(|V(D)||A(D)|) = O(n^2)$, since $|A(D)|=|E(G)|=O(n)$. This extended formulation automatically yields one for $\sub(G)$, since $\slack(G)$ and $\sub(G)$ are affinely equivalent. Finally, we get a size-$O(n^2)$ extended formulation for $\STAB(G)$ by adding the $2n$ inequalities $0 \leqslant x_v \leqslant 1$ for $v \in V(G)$, and invoking Lemma~\ref{lemStabVsSub}.
\qed
\end{proof}


\section{The general case}
\label{secGeneral_case}

In this section, we describe how Theorem~\ref{thm:main} can be proven using Theorems~\ref{thmStructure} and \ref{thm:EFcore}. Let $ G $ be a graph with $ \ocp(G) = 1 $. If $ \oct(G) \le 3 $, then $ \stab(G) $ has a linear-size extended formulation by Balas' theorem~\cite{Balas79}. Otherwise, $ \oct(G) \ge 4 $ and $ G $ can be decomposed as in Theorem~\ref{thmStructure}. In particular, $ G $ is the union of graphs $ H_0, T_1, \dots, T_\ell $ where $ H_0 $ has an even-face embedding in the projective plane and $ T_1, \dots, T_\ell $ are bipartite. Although the stable set polytopes of $ H_0,T_1,\dots,T_\ell $ admit small extended formulations and each $ T_i $ intersects $ H_0 \cup_{j \ne i} T_j $ in at most three nodes, it is not obvious how to obtain a small extended formulation for $ \stab(G) $. However, in some cases it is possible to use linear descriptions of the stable set polytopes of graphs $G_1,G_2$ to obtain a description of $ \stab(G_1 \cup G_2) $, provided that $G_1 \cap G_2$ has a specific structure, see \cite{Chvatal,ConfortiGP,BarahonaM}.

With this idea in mind, recall that not only $ H_0 $ but also the signed graph $ (H^+, \Sigma) $ has an even-face embedding in the projective plane. We will replace each signed clique used to define $ (H^+, \Sigma) $ by a constant size gadget $ H_i $ corresponding to each $ T_i $ in a way that the resulting graph $ G^{(\ell)} := H_0 \cup H_1 \cup \dots \cup H_\ell $ (the ``core'') still has an even-face embedding in the projective plane. Moreover, each $ T_i' := T_i \cup H_i $ will still be bipartite. In this way $ G $ is obtained from $ G^{(\ell)} $ by iteratively performing $ k $-sums with $ T_1',\dots,T_\ell' $ along $H_1, \dots, H_\ell$. In each such operation, the specific choice of the gadget will allow us to relate the extension complexities of the stable set polytopes of the participating graphs in a controlled way. Let us start with describing the gadgets that will be used.


\begin{definition} \label{defngadget}
A \emph{gadget} is a graph isomorphic to $P_3$, $P_4$, $S_{2,2,2}$ or $S_{2,3,3}$, see Figure~\ref{fig:gadgets}. Let $G$ be a graph with a linked $k$-separation $(G_0,G_1)$ such that $k \in \{2,3\}$ and $G_1$ is bipartite.  We say that a gadget $H$ is \emph{attachable} to $G_1$ (with respect to separation $(G_0,G_1)$) if its set of leaf nodes equals $V(G_0) \cap V(G_1)$, its set of non-leaf nodes is disjoint from $V(G)$, and $G_1 \cup H$ is bipartite.
\end{definition}

Note that if $G$ is a graph with a linked $k$-separation $(G_0,G_1)$ such that $k \in \{2,3\}$ and $G_1$ is bipartite, then there is a unique gadget $ H \in \{P_3, P_4, S_{2,2,2}, S_{2,3,3}\} $ that is attachable to $ G_1 $.

\begin{figure}[ht]
\centering
\begin{tikzpicture}[inner sep=2pt,scale=.5]
\tikzstyle{vtx}=[circle,draw,thick,fill=black!5]
\coordinate (v1) at (0:0);
\coordinate (v2) at (90:2);
\coordinate (v12) at ($ (v1) !.5! (v2) $);
\draw[thick] (v1) -- (v2);
\node[vtx] at (v1) {};
\node[vtx] at (v2) {};
\node[vtx] at (v12) {};
\draw (v1) ++(0.25,-.625) node {$P_3$};
\end{tikzpicture}
\qquad
\begin{tikzpicture}[inner sep=2pt,scale=.5]
\tikzstyle{vtx}=[circle,draw,thick,fill=black!5]
\coordinate (v1) at (0:0);
\coordinate (v2) at (90:2);
\coordinate (v112) at ($ (v1) !.333! (v2) $);
\coordinate (v122) at ($ (v2) !.333! (v1) $);
\draw[thick] (v1) -- (v2);
\node[vtx] at (v1) {};
\node[vtx] at (v2) {};
\node[vtx] at (v112) {};
\node[vtx] at (v122) {};
\draw (v1) ++(0.25,-.625) node {$P_4$};
\end{tikzpicture}
\qquad
\begin{tikzpicture}[inner sep=2pt,scale=.5]
\tikzstyle{vtx}=[circle,draw,thick,fill=black!5]
\coordinate (v0) at (0:0);
\coordinate (v1) at (90:2);
\coordinate (v2) at (210:2);
\coordinate (v3) at (330:2);
\coordinate (v01) at ($ (v0) !.5! (v1) $);
\coordinate (v02) at ($ (v0) !.5! (v2) $);
\coordinate (v03) at ($ (v0) !.5! (v3) $);
\draw[thick] (v0) -- (v1)  (v0) -- (v2)  (v0) -- (v3);
\node[vtx] at (v0) {};
\node[vtx] at (v1) {};
\node[vtx] at (v2) {};
\node[vtx] at (v3) {};
\node[vtx] at (v01) {};
\node[vtx] at (v02) {};
\node[vtx] at (v03) {};
\draw (v0) ++(0,-1.5) node {$S_{2,2,2}$};
\end{tikzpicture}
\qquad
\begin{tikzpicture}[inner sep=2pt,scale=.5]
\tikzstyle{vtx}=[circle,draw,thick,fill=black!5]
\coordinate (v0) at (0:0);
\coordinate (v1) at (90:2);
\coordinate (v2) at (210:2);
\coordinate (v3) at (330:2);
\coordinate (v01) at ($ (v0) !.5! (v1) $);
\coordinate (v002) at ($ (v0) !.333! (v2) $);
\coordinate (v022) at ($ (v0) !.666! (v2) $);
\coordinate (v003) at ($ (v0) !.333! (v3) $);
\coordinate (v033) at ($ (v0) !.666! (v3) $);
\draw[thick] (v0) -- (v1)  (v0) -- (v2)  (v0) -- (v3);
\node[vtx] at (v0) {};
\node[vtx] at (v1) {};
\node[vtx] at (v2) {};
\node[vtx] at (v3) {};
\node[vtx] at (v01) {};
\node[vtx] at (v002) {};
\node[vtx] at (v022) {};
\node[vtx] at (v003) {};
\node[vtx] at (v033) {};
\draw (v0) ++(0,-1.5) node {$S_{2,3,3}$};
\end{tikzpicture}
\caption{Gadgets and their names.} \label{fig:gadgets}
\end{figure}
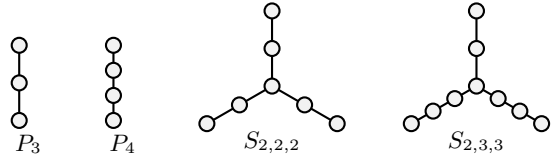

Next, let us formally describe how the signed cliques used to define $ (H^+, \Sigma) $ are replaced by gadgets in order to obtain the core.

\begin{definition} \label{defncore}
Let $ G $ be a $2$-connected graph with star structure $ H_0, T_1,\dots,T_\ell $. For each $i \in [\ell]$, pick a gadget $ H_i $ that is attachable to $ T_i $ with respect to the separation $(H_0 \cup \bigcup_{j \neq i} T_j, T_i)$. (We always assume that the set of non-leaf nodes of the gadgets $H_i$, $i \in [\ell]$ are mutually disjoint.) We call the graph  $ H_0 \cup H_1 \cup \dots \cup H_\ell$ the \emph{core}.
\end{definition}

\begin{lemma} \label{lem:combstructure}
    Every $2$-connected graph $ G $ with $ \ocp(G) = 1 $ and $ \oct(G) \ge 4 $ admits a star structure whose core has an even-face embedding in the projective plane.
\end{lemma}
\begin{proof}
The proof is immediate by choosing a star structure that satisfies Theorem~\ref{thmStructure}. \qed
\end{proof}

The remaining ingredient for our proof of Theorem~\ref{thm:main} will be the following result. To this end, let $(G_0,G_1)$ be a separation of graph $G$. Below, for $i \in \{0,1\}$, we call a vertex \emph{internal} if it belongs to $V(G_i) \setminus V(G_{1-i})$ and an edge of $G_i$ \emph{internal} if at least one of its ends is not in $G_{1-i}$.

\begin{theorem} \label{thm:xcbound}
Let $G$ be a $2$-connected, non-bipartite graph. Assume that $G$ has a $k$-separation $(G_0,G_1)$ such that $G_1$ is bipartite, and $k \in \{2,3\}$. Let $\mu_1$ denote the number of internal vertices and edges of $G_1$. Let $H$ be a gadget that is attachable to $G_1$, and let $G'_0 := G_0 \cup H$. Then
$$
\xc(\STAB(G)) \leqslant \xc(\STAB(G'_0)) + O(\mu_1)\,.
$$
\end{theorem}

Before we continue with the proof of Theorem~\ref{thm:xcbound} in the next section, let us see how this yields a proof of our main result. 
\begin{proof}[Proof of Theorem~\ref{thm:main}]
By induction on the number of nodes $n$, we may assume that $G$ is $2$-connected. Indeed, suppose that $G$ has a $k$-separation $(G_0,G_1)$ with $k \in \{0,1\}$. For $i \in \{0,1\}$, let $n_i := |V(G_i)|$. Thus $n = n_0 + n_1 - k$. If $c$ is any constant such that $\xc(\STAB(G_i)) \leqslant c \cdot n_i^2$ for $i \in \{0,1\}$, we get
$$
\xc(\STAB(G)) \leqslant \xc(\STAB(G_0)) + \xc(\STAB(G_1)) \leqslant c \cdot n_0^2 + c \cdot n_1^2 \leqslant c \cdot n^2\,.
$$

As observed above, if $\oct(G) \leqslant 3$ then $\STAB(G)$ trivially has a size-$O(n^2)$ extended formulation. Now assume that $\ocp(G) = 1$ and $\oct(G) \geqslant 4$. Let $H_0, T_1, \ldots, T_\ell$ be a star structure of $G$ as in Lemma~\ref{lem:combstructure}. Since $G$ is $2$-connected, each separation $(H_0 \cup_{j \neq i} T_j, T_i)$ is either a $2$- or a $3$-separation. For each $i \in [\ell]$, we consider the graph
$$
G^{(i)} := H_0 \cup H_1 \cup \cdots \cup H_{i} \cup T_{i+1} \cup \dots \cup T_{\ell}\,.
$$
where $H_i$ denotes a gadget attachable to $T_i$. For $i \in [\ell]$, let $\mu_i$ denote the number of internal vertices and edges of $T_i$. Notice that $G^{(\ell)}$ is the core, and thus by Lemma \ref{lem:combstructure} has an even-face embedding in the projective plane. By Theorem~\ref{thm:xcbound},
$$
\xc(\STAB(G^{(i-1)})) \leqslant \xc(\STAB(G^{(i)})) + O(\mu_i)\,.
$$
Since $|V(G^{(\ell)})| = O(n)$, Theorem~\ref{thm:EFcore} implies $\xc(\STAB(G^{(\ell)})) = O(n^2)$. Since moreover $\sum_{i=1}^\ell \mu_i \leqslant |V(G)| + |E(G)| = O(n^2)$, we have
$$
\xc(\STAB(G)) = 
\xc(\STAB(G^{(0)}))
\leqslant
\xc(\STAB(G^{(\ell)})) + O\left(\sum_{i=1}^\ell \mu_i \right) = O(n^2)\,. \eqno \qed
$$
\end{proof}

\section{Extended formulation for small separations}
\label{secProofEF}

In this section we describe an extended formulation that yields the bound claimed in Theorem~\ref{thm:xcbound}. Given a stable set $S$ in a graph $G$, we say that an edge is \emph{slack} if neither of its ends is in $S$. We denote by $\sigma(S)$ the set of slack edges, or $\sigma_G(S)$ should the graph not be clear from the context. An edge is said to be \emph{tight} if it is not slack.

\begin{lemma} \label{lem:EF}
Let $G$, $G_0$, $G_1$ and $H$ be as in Theorem~\ref{thm:xcbound}. Letting $\overline{\STAB}(G_1')$ denote the convex hull of characteristic vectors of stable sets $S$ in $G_1'$ having at most one slack edge in $H$, we have
\begin{equation} \label{eq:EF}
\STAB(G) = \{ (x^0, x^1, x^{01}) \in \R^{V(G)} \mid \exists x^H : 
\begin{array}[t]{l}
(x^0, x^{01}, x^H) \in \stab(G_0'),\\[1ex]
(x^1, x^{01}, x^H) \in \overline{\STAB}(G_1') \}.
\end{array}
\end{equation}
where $ x^0 \in \R^{V(G_0) \setminus V(G_1)} $, $ x^1 \in \R^{V(G_1) \setminus V(G_0)} $, $ x^{01} \in \R^{V(G_0) \cap V(G_1)} $ and $ x^H \in \R^{V(H) \setminus V(G)} $.
\end{lemma}

Let us first verify that Lemma~\ref{lem:EF} indeed implies Theorem~\ref{thm:xcbound}.

\begin{proof}[Proof of Theorem~\ref{thm:xcbound}]
By Lemma~\ref{lem:EF}, we have
\[
    \xc(\stab(G)) \le \xc(\stab(G_0')) + \xc(\overline{\STAB}(G_1'))\,.
\]
Since gadget $H$ has constant size, $ \overline{\STAB}(G_1') $ is the convex hull of the union of a constant number of faces of $ \stab(G_1') $ in which the coordinates of the nodes in $ H $ are fixed. Hence by Balas' union of polytopes~\cite{Balas79}, we obtain $\xc(\overline{\STAB}(G_1')) = O(\xc(\stab(G_1))) = O(|V(G_1)| + |E(G_1)|)$. Since $|V(G_1)| + |E(G_1)| - \mu_1 \le 6 $ and $\mu_1 \ge 1$, we conclude
\[
    \xc(\overline{\STAB}(G_1')) = O(\mu_1)\,.
\]
This proves the claim. \qed
\end{proof}

In the proof of Lemma~\ref{lem:EF} we will exploit that the facets of stable set polytopes have a special structure, which we describe next.

\subsection{Reducing to edge-induced weights}

We call a weight function $w : V(G) \to \R$ on the nodes of $G$ \emph{edge-induced} if there is a nonnegative cost function $c : E(G) \to \R_{\ge 0}$ such that $w(v) = c(\delta(v))$ for all $v \in V(G)$. For a given node-weighted graph $(G,w)$ we let $\alpha(G,w)$ denote the maximum weight of a stable set.

\begin{lemma} \label{lem:edge_induced_weights}
Let $G = (V,E)$ be a graph without isolated nodes and let $w : V \to \R$ be a weight function. There exists an edge-induced weight function $w' : V \to \R$ such that $w(v) \leqslant w'(v)$ for all nodes $v$ and $\alpha(G,w) = \alpha(G,w')$. In particular, the node weights of every non-trivial facet-defining inequality of $\STAB(G)$ are edge-induced.
\end{lemma}
\begin{proof} Let $x^*$ denote an optimal solution of the  LP $\max \{\sum_{v \in V} w(v) x_v \mid x_v + x_w \leqslant 1\ \forall vw \in E,\ x \geqslant \mathbf{0}\}$  and $y^*$ be an optimal solution of its dual  $\min \{\sum_{e \in E} y_e \mid y(\delta(v)) \geqslant w(v)\ \forall v \in V,\ y \geqslant \mathbf{0}\}$.
 
Consider the weight function $w'$ such that $w'(v) := y^*(\delta(v))$. Clearly, $w'(v) \ge w(v)$ for all nodes $v$ and  $w'$ is edge-induced.  Consider the above  LPs where $w'$ replaces $w$. Then $x^*$ and $y^*$ remain optimal solutions as they are feasible and satisfy complementary slackness. Moreover the values of the new LPs remain unchanged, as the objective function of the dual is not changed.

Let $V_0 := \{v \in V \mid x^*_v = 0\}$. Since  $w(v) = y^*(\delta(v))$ for all $v \in V \setminus V_0$ by complementary slackness, $w(v)=w'(v)$ for all $v \in V \setminus V_0$.
We have 
$$
\alpha(G,w) \leqslant \alpha(G,w') = \alpha(G - V_0,w') = \alpha(G - V_0, w) \leqslant \alpha(G,w)\,.
$$  
Above, the first inequality follows from $w \leqslant w'$.  The first equality follows from a result of Nemhauser and Trotter \cite{NT74}. Their result implies that $(G,w')$ has a maximum weight stable set disjoint from $V_0$. The second equality follows from the fact that $w(v) = w'(v)$ for all $v \in V \setminus V_0$. Hence, equality holds throughout and $\alpha(G,w) = \alpha(G,w')$. 

Finally, if $\sum_{v \in V} w(v) x_v \le \alpha(G,w)$ induces a non-trivial facet of $\STAB(G)$, there cannot exist $w' \neq w$ such that $w'\ge w$ and $\alpha(G,w')=\alpha(G,w)$. Hence the above argument shows that the node weights of every non-trivial facet-defining inequality of $\STAB(G)$ are edge-induced. 
\qed
\end{proof}

In fact, in~\cite[Propositions 11 and 14]{CFHJW19} it is shown that one can optimize over $\slack(G) = \sigma(\sub(G))$ instead of $\sigma(\STAB(G))$ without changing the optimum. However, we will not need this here.  

For $c : E(G) \to \R_{\ge 0}$, we let 
\begin{equation}
\label{eq:beta_def}
\beta(G,c) := \min \left\{\sum_{e \in E(G)} c(e) y_e \mid y \in \sigma(\STAB(G))\right\}\,.
\end{equation}
Finally, we need the following easy lemma.

\begin{lemma} \label{lem:OPT_STAB_sub_slack}
Let $G = (V,E)$ be a graph. If $w : V(G) \to \R$ is induced by $c : E(G) \to \R_{\ge 0}$, then $\alpha(G,w) = c(E(G))- \beta(G,c)$.
\end{lemma}

\begin{proof}
A \emph{star} of $G$ is a set of edges of the form $\delta(v)$, for some $v \in V(G)$.  Note that since $w : V(G) \to \R$ is induced by $c : E(G) \to \R_{\ge 0}$, 
\begin{equation*}
\alpha(G,w)=\max\{c(F) \mid \text{$F$ is the edge-disjoint union of stars}\}= c(E(G))- \beta(G,c). \qed
\end{equation*}

\end{proof}

\subsection{Correctness of the extended formulation}
\label{subsecProofLemEF}

In this section we prove Lemma~\ref{lem:EF}. To this end, let $R(G)$ denote the right-hand side of \eqref{eq:EF}. Notice that for each stable set $S$ of $G$, there exists a stable set $S'$ of $G' := G \cup H$ such that $S' \cap V(G) = S$ and moreover at most one edge of $H$ is slack with respect to $S'$. The inclusion $\STAB(G) \subseteq R(G)$ follows directly from this.

In order to prove the reverse inclusion $R(G) \subseteq \STAB(G)$, first observe that $ R(G) \subseteq \R^{V(G)}_{\ge 0} $. Thus, by Lemma~\ref{lem:edge_induced_weights} it suffices to show that, for all edge-induced node weights $w : V(G) \to \mathbb{R}$, the inequality
\begin{equation}
\label{eq:target_node-space}
\sum_{v \in V(G)} w(v) x_v \leqslant \alpha(G,w)
\end{equation}
is valid for all $ x \in R(G) $. As in Section~\ref{secProjective_planar_case} it will be convenient to work in the edge space instead of the node space. To this end, let $c : E(G) \to \mathbb{R}_+$ be non-negative edge costs, and let $w(v) := c(\delta(v))$ for every node $v$. By Lemma~\ref{lem:OPT_STAB_sub_slack} we see that \eqref{eq:target_node-space} is valid for $ R(G) $ if and only if
\begin{equation}
\label{eq:target}
\sum_{e \in E(G)} c(e) y_e \geqslant \beta(G,c)
\end{equation}
is satisfied by all points $ y \in \sigma(R(G)) $. Our proof strategy to obtain \eqref{eq:target} is to seek additional costs $c^H : E(H) \to \R_{\ge 0}$ such that
\begin{equation}
\label{eq:master}
\sum_{e \in E(G_0)} c(e) y^0_e + \sum_{e \in E(H)} c^H(e) y^H_e \geqslant \beta(G,c)
\end{equation}
is valid for all $(y^0,y^H) \in \sigma(\STAB(G'_0))$ and
\begin{equation}
\label{eq:servant}
\sum_{e \in E(G_1)} c(e) y^1_e - \sum_{e \in E(H)} c^H(e) y^H_e \geqslant 0
\end{equation}
is valid for all $(y^1,y^H) \in \sigma(\overline{\STAB}(G'_1))$.

We claim that this will yield \eqref{eq:target}. Indeed, for every vector $ y = (y^0, y^1) \in \sigma(R(G)) $ there exists a vector $ y^H $ (the image of $ (x^{01}, x^H) $ under $ \sigma_H $) with $ (y^0, y^H) \in \sigma(\stab(G_0')) $ and $ (y^1, y^H) \in \sigma(\overline{\stab}(G'_0))$. This implies that the inequalities in \eqref{eq:master} and \eqref{eq:servant} are satisfied. Now \eqref{eq:target} follows since it is the sum of these two inequalities.

Let us first focus on Inequality~\eqref{eq:servant}.
Independently of how the edge costs $c^H$ are defined, in order to prove that it holds for all $(y^0,y^H) \in \sigma(\overline{\STAB}(G'_1))$, we may assume that $y^H$ is a 0/1-vector with at most one nonzero entry. The general case follows by convexity. For $F \subseteq E(H)$, we let $\chi^F$ be the vector in $\{0,1\}^{E(H)}$ such that $\chi_e^F=1$ if and only if $e \in F$.  Since the case $y^H = \mathbf{0}$ is trivial, assume that $y^H = \chi^{\{f\}}$ for some $f \in E(H)$. Hence \eqref{eq:servant} can be rewritten as
\begin{equation}
\label{eq:servant_bis}
\sum_{e \in E(G_1)} c(e) y^1_e \geqslant c^H(f)\,.
\end{equation}
This suggests the following definition of $c^H$. For $F \subseteq E(H)$, we let
\[
\gamma(F) := \min \left\{c\big(\sigma(S) \cap E(G_1)\big) \mid S \textrm{ stable set of } G'_1,\ \sigma(S) \cap E(H) = F\right\} \in \R_{\ge 0} \cup \{\infty\}\,.
\]
We say that $F$ is \emph{feasible} if $\gamma(F)$ is finite, that is, there exists a stable set $S$ of $G'_1$ such that $\sigma(S) \cap E(H) = F$. Notice that $F := \{f\}$ is feasible for all $f \in E(H)$. By setting $c^H(f) := \gamma(\{f\}) \in \R_{\ge 0}$ for each $f \in E(H)$ we clearly satisfy \eqref{eq:servant_bis}, and hence \eqref{eq:servant} is valid for all $ (y^1,y^H) \in \sigma(\overline{\STAB}(G'_1))$ for this choice of $ c^H $.

It remains to prove that with this choice of $ c^H $ the inequality in \eqref{eq:master} is valid for all $(y^0,y^H) \in \sigma(\STAB(G'_0))$. To this end, we need the following two observations.

\begin{lemma} \label{lem:characterizations_gamma}
Let $G_1$ and $H$ be as in Theorem~\ref{thm:xcbound}, and let $G'_1 := G_1 \cup H$. Hence, $G'_1$ is bipartite. Let $c : E(G_1) \to \R_{\ge 0}$ be nonnegative edge costs. Assume that $F \subseteq E(H)$ is feasible. Letting $x$ and $y = (y^1,y^H)$ denote arbitrary points in $\R^{V(G'_1)}$ and $\R^{E(G'_1)}$ respectively, and letting $M$ denote the incidence matrix of $G'_1$, consider the following LPs:
\begin{align*}
\LP_1(F) &:= \min \left\{\sum_{e \in E(G_1)} c(e) y^1_e \mid Mx + y = \mathbf{1},\ y \geqslant \mathbf{0},\ y^H = \chi^F,\ x \geqslant \mathbf{0}\right\} \quad \textrm{and}\\
\LP_2(F) &:= \min \left\{\sum_{e \in E(G_1)} c(e) y^1_e \mid Mx + y = \mathbf{1},\ y \geqslant \mathbf{0},\ y^H = \chi^F\right\}\,. 
\end{align*}
Then $\gamma(F) = \LP_1(F) = \LP_2(F)$.
\end{lemma}

\begin{proof}
That $\gamma(F) = \LP_1(F)$ follows directly from the fact that $G'_1$ is bipartite. Furthermore, it is clear that $\LP_1(F) \geqslant \LP_2(F)$. If $F$ is empty, then $\LP_1(F) = \LP_2(F) = 0$ since $(x, y) :=
(\frac{1}{2} \mathbf{1}, \mathbf{0})$ is optimal for both LPs. From now on, assume that $F$ is nonempty, and let $v_0 \in V(H)$ be any node that is incident to some edge of $F$. 

Now consider the LP obtained from $\LP_3(F)$ by adding the constraint $x_{v_0} = 0$:
$$
\LP_3(F) := \min \left\{\sum_{e \in E(G_1)} c(e) y^1_e \mid Mx + y = \mathbf{1},\ y \geqslant \mathbf{0},\ y^H = \chi^F,\ x_{v_0} = 0\right\}\,. 
$$
Since $G'_1$ is bipartite, $\LP_2(F) = \LP_3(F)$ since adding the extra constraint does not change the set of feasible $y$ vectors. Thanks to the extra constraint, the feasible region of $\LP_3(F)$ is pointed.

Consider an extreme optimal solution $(\bar{x},\bar{y})$ of $\LP_3(F)$. Since $M$ is totally unimodular, we may assume that both $\bar{x}$ and $\bar{y}$ are integral. Since $F$ is feasible, $\bar{x}_v \in \{0,1\}$ for all $v \in V(H)$. We claim that $(\bar{x},\bar{y})$ is feasible for $\LP_1(F)$. Observe that the claim implies $\LP_1(F) \leqslant \LP_3(F) = \LP_2(F)$ and thus $\LP_1(F) = \LP_2(F)$. 

If $\bar{x}$ is nonnegative, we are done. Otherwise, we can find disjoint sets $V_{\alpha}$ and $V_{1-\alpha}$ for some $\alpha \in \Z_{<0}$ such that $\bar{x}_v = \alpha$ for all $v \in V_{\alpha}$, $\bar{x}_v = 1 - \alpha$ for all $v \in V_{1-\alpha}$ and no edge $e$ with $y_e = 0$ has exactly one end in $V_{\alpha} \cup V_{1-\alpha}$. Since $\alpha < 0$ and $1 - \alpha > 1$, we see that both $V_{\alpha}$ and $V_{1-\alpha}$ are disjoint from $V(H)$. Let $\bar{x}' := \bar{x} + \chi^{V_{\alpha}} - \chi^{V_{1-\alpha}}$, $\bar{y}' := \mathbf{1} - M\bar{x}'$, $\bar{x}'' := \bar{x} - \chi^{V_{\alpha}} + \chi^{V_{1-\alpha}}$ and $\bar{y}'' := \mathbf{1} - M\bar{x}''$. Both $(\bar{x}',\bar{y}')$ and $(\bar{x}'',\bar{y}'')$ are feasible for $\LP_3(F)$, contradicting the extremality of $(\bar{x},\bar{y})$. \qed
\end{proof}

\begin{lemma} \label{lem:subadditivity}
If $F \subseteq E(H)$ is feasible and the disjoint union of $A$ and $B$, then $\gamma(F) \leqslant \gamma(A) + \gamma(B)$.
\end{lemma}
\begin{proof}
We may assume that $A$ and $B$ are both feasible, otherwise there is nothing to prove. Let $(x,y)$ and $(z,t)$ be optimal solutions of $\LP_2(A)$ and $\LP_2(B)$ respectively (see~Lemma~\ref{lem:characterizations_gamma}). If we let $u := x + z - \frac{1}{2} \mathbf{1}$ and $v := y + t$, then $(u,v)$ is feasible for $\LP_2(F)$ since
$$
Mu + v = Mx + Mz - \frac{1}{2} M\mathbf{1} + v = (\mathbf{1} - y) + (\mathbf{1} - t) - \mathbf{1} + v = \mathbf{1}\,,
$$
$v \geqslant \mathbf{0}$ and $\chi^A + \chi^B = \chi^F$. By Lemma~\ref{lem:characterizations_gamma}, this shows that $\gamma(F) \leqslant \gamma(A) + \gamma(B)$.
\qed
\end{proof}

To prove that the inequality in \eqref{eq:master} is valid for all $(y^0,y^H) \in \sigma(\STAB(G'_0))$, it suffices to consider any vertex $(y^0,y^H)$ of $\sigma(\STAB(G'_0))$ minimizing the left-hand size of~\eqref{eq:master}. We may even assume that $(y^0,y^H)$ minimizes $||y^H||_1$ among all such vertices.

Let $S^0$ denote the stable set of $G'_0$ corresponding to $(y^0,y^H)$ and let $F := \sigma(S^0) \cap E(H)$. Note that $y^H = \chi^F$. Observe that $ S^0 $ is not properly contained in another stable set, since this would contradict the minimality of $ y $. Moreover, we claim that $F$ has at most one edge. In order to prove the claim, we consider only the case where $H = S_{2,2,2}$, see Figure~\ref{fig:gadgets}. The other cases are easier or similar, and we leave the details to the reader.

Let us assume that $ F $ contains at least two edges, that is, $ \|y^H\|_1 \ge 2 $. We will replace $y^H$ by a new vector $\bar{y}^H \in \{0,1\}^{E(H)}$ such that $(y^0,\bar{y}^H) \in \sigma(\STAB(G'_0))$ with smaller $\ell_1$-norm in such a way that the cost of $(y^0,\bar{y}^H)$ is not higher than that of $(y^0,y^H)$, arriving at a contradiction. In order to prove that $(y^0,\bar{y}^H) \in \sigma(\STAB(G'_0))$ we will explain how to obtain the corresponding stable set $\bar{S}^0$ from stable set $S^0$ in each case. To guarantee that the cost of $(y^0,\bar{y}^H)$ does not exceed that of $(y^0,y^H)$, we will mainly rely on Lemma~\ref{lem:subadditivity}.

To distinguish the different cases, let $v_1$, $v_2$ and $v_3$ denote the leaves of $H$ and $v_0$ denote its degree-$3$ node. For $i, j \in \{0,1,2,3\}$ we let $P_{ij}$ denote the $v_i$--$v_j$ path in $H$. For $i \in [3]$, let $v_{0i}$ denote the middle vertex of $P_{ij}$ and let $e_{i}$ and $f_{i}$ denote the edges of the path $P_{0i}$ incident to $v_i$ and $v_0$ respectively. The relevant cases and the replacements are listed in Figure~\ref{fig:replacements}. We treat each of them below. Notice that the case $|S^0 \cap \{v_1,v_2,v_3\}| = 3$ cannot arise since this would contradict the maximality of $S^0$.\medskip

\noindent \emph{Case 1: $ |S^0 \cap \{v_1,v_2,v_3\}| = 0 $.} In this case we set $ \bar{y}^H := \zerovec $, which corresponds to letting $ \bar{S}^0 := (S^0 \cup \{v_{01}, v_{02}, v_{03}\}) \setminus \{v_0\} $. In this case it is clear that the cost of $(y^0,\bar{y}^H)$ is at most the cost of $(y^0,y^H)$.\smallskip

\noindent \emph{Case 2: $ |S^0 \cap \{v_1,v_2,v_3\}| = 1 $.} We may assume that $ S^0 \cap \{v_1,v_2,v_3\} = \{v_3\} $. Since $ |F| \ge 2 $ and $ S^0 $ is maximal, we must have $ S^0 \cap V(H) = \{ v_0, v_3 \} $ and hence $ y^H = \chi^{\{e_1, e_2\}} $. 

We let $\bar{y}^H := \chi^{\{f_3\}}$, which corresponds to letting $\bar{S}^0 := S^0 \setminus \{v_0\} \cup \{v_{01},v_{02}\}$. The cost of $(y^0,\bar{y}^H)$ equals the cost of $(y^0,y^H)$ minus $\gamma(\{e_1\}) + \gamma(\{e_2\}) - \gamma(\{f_3\}) = \gamma(\{e_1\}) + \gamma(\{e_2\}) - \gamma(\{e_1,e_2\}) \geqslant 0$. The equality follows from the fact that stable sets $S$ of $G'_1$ such that $\sigma(S) \cap E(H) = \{f_3\}$ and stable sets $S$ of $G'_1$ such that $\sigma(S) \cap E(H) = \{e_1,e_2\}$ have the same intersection with the leaves of $H$. The inequality follows from Lemma~\ref{lem:subadditivity}.\smallskip

\noindent \emph{Case 3: $ |S^0 \cap \{v_1,v_2,v_3\}| = 2 $.} We may assume that $ S^0 \cap \{v_1,v_2,v_3\} = \{v_1,v_2\} $. Again, since $ |F| \ge 2 $ and $ S^0 $ is maximal, we must have $ S^0 \cap V(H) = \{ v_1,v_2,v_{03} \} $ and hence $ y^H = \chi^{\{f_1, f_2\}} $. We let $\bar{y}^H := \chi^{\{e_{3}\}}$, which corresponds to letting $\bar{S}^0 := S^0 \setminus \{v_{03}\} \cup \{v_{01},v_{02}\}$. Similar to the previous case, we obtain that the cost of $(y^0,\bar{y}^H)$ equals the cost of $(y^0,y^H)$ minus $\gamma(\{f_1\}) + \gamma(\{f_2\}) - \gamma(\{e_3\}) = \gamma(\{f_1\}) + \gamma(\{f_2\}) - \gamma(\{f_1,f_2\}) \geqslant 0$. \smallskip

Thus, $ F $ has indeed at most one edge. There exists a stable set $S^1$ of $G'_1$ that is a minimizer for $\gamma(F)$ such that $S^1 \cap V(G) \cap V(H) = S^0 \cap V(G) \cap V(H)$. Hence, $S := S^1 \cup S^0$ is a stable set of $G$. Let $(y^0,y^1)$ denote the characteristic vector of $\sigma(S)$, so that $(y^0,y^1) \in \sigma(\STAB(G))$. We get
\begin{align*}
\sum_{e \in E(G_0)} c(e) y^0_e + \sum_{e \in E(H)} c^H(e) y^H_e
& = \sum_{e \in E(G_0)} c(e)y^0_e + \gamma(F) \\
& = \sum_{e \in E(G_0)} c(e)y^0_e + \sum_{e \in E(G_1)} c(e)y^1_e
\geqslant \beta(G,c)\,.
\end{align*}
Above, the first equality comes from the fact that $F$ has at most one edge, the definition of $c^H(f)$ for $f \in E(H)$ and $\gamma(\emptyset) = 0$. The second equality follows from the hypothesis that $S^1$ is a minimizer for $\gamma(F)$. Finally, the inequality is due to the validity of~\eqref{eq:target} for $\sigma(\STAB(G))$. This shows that \eqref{eq:master} is indeed valid for $(y^0,y^H) \in \sigma(\STAB(G'_0))$, which concludes the proof of Lemma~\ref{lem:EF}.

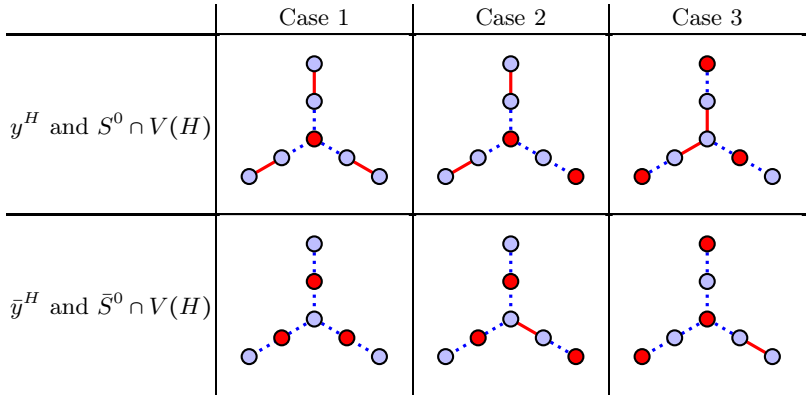
\begin{figure}[ht] \centering
\begin{tabular}{r|c|c|c}
& Case 1& Case 2& Case 3\\
\hline
\mbox{}\raise7ex\hbox{$y^H$ and $S^0 \cap V(H)$}\quad
&\begin{tikzpicture}[inner sep=2pt,scale=.5]
\tikzstyle{vtx}=[circle,draw,thick,fill=black!5]
\tikzstyle{vtx0}=[circle,draw,thick,fill=blue!25!white]
\tikzstyle{vtx1}=[circle,draw,thick,fill=red]
\draw[white] (-2.5,2.75) rectangle (2.5,-1.75);
\coordinate (v0) at (0:0);
\coordinate (v1) at (90:2);
\coordinate (v2) at (210:2);
\coordinate (v3) at (330:2);
\coordinate (v01) at ($ (v0) !.5! (v1) $);
\coordinate (v02) at ($ (v0) !.5! (v2) $);
\coordinate (v03) at ($ (v0) !.5! (v3) $);
\draw[very thick, red] (v01) -- (v1);
\draw[very thick, red] (v02) -- (v2);
\draw[very thick, red] (v03) -- (v3);
\draw[very thick, blue, dotted] (v0) -- (v03) (v01) -- (v0) -- (v02);
\node[vtx1] at (v0) {};
\node[vtx0] at (v1) {};
\node[vtx0] at (v2) {};
\node[vtx0] at (v3) {};
\node[vtx0] at (v01) {};
\node[vtx0] at (v02) {};
\node[vtx0] at (v03) {};
\end{tikzpicture}
&\begin{tikzpicture}[inner sep=2pt,scale=.5]
\tikzstyle{vtx}=[circle,draw,thick,fill=black!5]
\tikzstyle{vtx0}=[circle,draw,thick,fill=blue!25!white]
\tikzstyle{vtx1}=[circle,draw,thick,fill=red]
\draw[white] (-2.5,2.75) rectangle (2.5,-1.75);
\coordinate (v0) at (0:0);
\coordinate (v1) at (90:2);
\coordinate (v2) at (210:2);
\coordinate (v3) at (330:2);
\coordinate (v01) at ($ (v0) !.5! (v1) $);
\coordinate (v02) at ($ (v0) !.5! (v2) $);
\coordinate (v03) at ($ (v0) !.5! (v3) $);
\draw[very thick, red] (v01) -- (v1);
\draw[very thick, red] (v02) -- (v2);
\draw[very thick, blue, dotted] (v0) -- (v3) (v01) -- (v0) -- (v02);
\node[vtx1] at (v0) {};
\node[vtx0] at (v1) {};
\node[vtx0] at (v2) {};
\node[vtx1] at (v3) {};
\node[vtx0] at (v01) {};
\node[vtx0] at (v02) {};
\node[vtx0] at (v03) {};
\end{tikzpicture}
&\begin{tikzpicture}[inner sep=2pt,scale=.5]
\tikzstyle{vtx}=[circle,draw,thick,fill=black!5]
\tikzstyle{vtx0}=[circle,draw,thick,fill=blue!25!white]
\tikzstyle{vtx1}=[circle,draw,thick,fill=red]
\draw[white] (-2.5,2.75) rectangle (2.5,-1.75);
\coordinate (v0) at (0:0);
\coordinate (v1) at (90:2);
\coordinate (v2) at (210:2);
\coordinate (v3) at (330:2);
\coordinate (v01) at ($ (v0) !.5! (v1) $);
\coordinate (v02) at ($ (v0) !.5! (v2) $);
\coordinate (v03) at ($ (v0) !.5! (v3) $);
\draw[very thick, red] (v01) -- (v0) -- (v02);
\draw[very thick, blue, dotted] (v1) -- (v01) (v2) -- (v02) (v0) -- (v3);
\node[vtx0] at (v0) {};
\node[vtx1] at (v1) {};
\node[vtx1] at (v2) {};
\node[vtx0] at (v3) {};
\node[vtx0] at (v01) {};
\node[vtx0] at (v02) {};
\node[vtx1] at (v03) {};
\end{tikzpicture}
\\
\hline
\mbox{}\raise7ex\hbox{$\bar{y}^H$ and $\bar{S}^0 \cap V(H)$}\quad
&\begin{tikzpicture}[inner sep=2pt,scale=.5]
\tikzstyle{vtx}=[circle,draw,thick,fill=black!5]
\tikzstyle{vtx0}=[circle,draw,thick,fill=blue!25!white]
\tikzstyle{vtx1}=[circle,draw,thick,fill=red]
\draw[white] (-2.5,2.75) rectangle (2.5,-1.75);
\coordinate (v0) at (0:0);
\coordinate (v1) at (90:2);
\coordinate (v2) at (210:2);
\coordinate (v3) at (330:2);
\coordinate (v01) at ($ (v0) !.5! (v1) $);
\coordinate (v02) at ($ (v0) !.5! (v2) $);
\coordinate (v03) at ($ (v0) !.5! (v3) $);
\draw[very thick, blue, dotted] (v1) -- (v0) -- (v2) (v0) -- (v3);
\node[vtx0] at (v0) {};
\node[vtx0] at (v1) {};
\node[vtx0] at (v2) {};
\node[vtx0] at (v3) {};
\node[vtx1] at (v01) {};
\node[vtx1] at (v02) {};
\node[vtx1] at (v03) {};
\end{tikzpicture}
&
\begin{tikzpicture}[inner sep=2pt,scale=.5]
\tikzstyle{vtx}=[circle,draw,thick,fill=black!5]
\tikzstyle{vtx0}=[circle,draw,thick,fill=blue!25!white]
\tikzstyle{vtx1}=[circle,draw,thick,fill=red]
\draw[white] (-2.5,2.75) rectangle (2.5,-1.75);
\coordinate (v0) at (0:0);
\coordinate (v1) at (90:2);
\coordinate (v2) at (210:2);
\coordinate (v3) at (330:2);
\coordinate (v01) at ($ (v0) !.5! (v1) $);
\coordinate (v02) at ($ (v0) !.5! (v2) $);
\coordinate (v03) at ($ (v0) !.5! (v3) $);
\draw[very thick, red] (v0) -- (v03);
\draw[very thick, blue, dotted] (v1) -- (v0) -- (v2) (v03) -- (v3);
\node[vtx0] at (v0) {};
\node[vtx0] at (v1) {};
\node[vtx0] at (v2) {};
\node[vtx1] at (v3) {};
\node[vtx1] at (v01) {};
\node[vtx1] at (v02) {};
\node[vtx0] at (v03) {};
\end{tikzpicture}
&\begin{tikzpicture}[inner sep=2pt,scale=.5]
\tikzstyle{vtx}=[circle,draw,thick,fill=black!5]
\tikzstyle{vtx0}=[circle,draw,thick,fill=blue!25!white]
\tikzstyle{vtx1}=[circle,draw,thick,fill=red]
\draw[white] (-2.5,2.75) rectangle (2.5,-1.75);
\coordinate (v0) at (0:0);
\coordinate (v1) at (90:2);
\coordinate (v2) at (210:2);
\coordinate (v3) at (330:2);
\coordinate (v01) at ($ (v0) !.5! (v1) $);
\coordinate (v02) at ($ (v0) !.5! (v2) $);
\coordinate (v03) at ($ (v0) !.5! (v3) $);
\draw[very thick, red] (v03) -- (v3);
\draw[very thick, blue, dotted] (v03) -- (v0) (v1) -- (v0) -- (v2);
\node[vtx1] at (v0) {};
\node[vtx1] at (v1) {};
\node[vtx1] at (v2) {};
\node[vtx0] at (v3) {};
\node[vtx0] at (v01) {};
\node[vtx0] at (v02) {};
\node[vtx0] at (v03) {};
\end{tikzpicture}
\end{tabular}

\caption{Replacements in the proof of Lemma~\ref{lem:EF} (top row: before, bottom row: after). Red thick edges are slack. Blue thick, dotted edges are tight. Red nodes are in the stable set, blue nodes are not. \label{fig:replacements}}
\end{figure}

\section*{Acknowledgements}

This paper was supported by ERC Consolidator Grant 615640-ForEFront. We would like to thank the IPCO reviewers for their comments. Also, we would like to mention that the total unimodularity test of Truemper \& Walter~\cite{WalterTruemper} provided valuable insights in the early stages of our project, especially for gaining a better understanding of the algorithm in~\cite{AWZ17}.

\bibliography{references}{}

\begin{thebibliography}{10}
\providecommand{\url}[1]{\texttt{#1}}
\providecommand{\urlprefix}{URL }
\providecommand{\doi}[1]{https://doi.org/#1}

\bibitem{AWZ17}
Artmann, S., Weismantel, R., Zenklusen, R.: A strongly polynomial algorithm for
  bimodular integer linear programming. In: S{TOC}'17---{P}roceedings of the
  49th {A}nnual {ACM} {SIGACT} {S}ymposium on {T}heory of {C}omputing, pp.
  1206--1219. ACM, New York (2017)

\bibitem{Balas79}
Balas, E.: Disjunctive programming. Ann. Discrete Math.  \textbf{5},  3--51
  (1979), discrete optimization (Proc. Adv. Res. Inst. Discrete Optimization
  and Systems Appl., Banff, Alta., 1977), II

\bibitem{BarahonaM}
Barahona, F., Mahjoub, A.R.: Compositions of graphs and polyhedra ii: stable
  sets. SIAM Journal on Discrete Mathematics  \textbf{7}(3),  359--371 (1994)

\bibitem{BDEHN14}
Bonifas, N., Di~Summa, M., Eisenbrand, F., H{\"a}hnle, N., Niemeier, M.: On
  sub-determinants and the diameter of polyhedra. Discrete \& Computational
  Geometry  \textbf{52}(1),  102--115 (2014)

\bibitem{CevallosWeltgeZenklusen}
Cevallos, A., Weltge, S., Zenklusen, R.: Lifting linear extension complexity
  bounds to the mixed-integer setting. In: Proceedings of the Twenty-Ninth
  Annual ACM-SIAM Symposium on Discrete Algorithms. pp. 788--807. SODA '18,
  Society for Industrial and Applied Mathematics, Philadelphia, PA, USA (2018),
  \url{http://dl.acm.org/citation.cfm?id=3174304.3175321}

\bibitem{Chvatal}
Chv{\'a}tal, V.: On certain polytopes associated with graphs. Journal of
  Combinatorial Theory, Series B  \textbf{18}(2),  138--154 (1975)

\bibitem{CFHJW19}
Conforti, M., Fiorini, S., Huynh, T., Joret, G., Weltge, S.: The stable set
  problem in graphs with bounded genus and bounded odd cycle packing number.
  In: Proceedings of the Fourteenth Annual ACM-SIAM Symposium on Discrete
  Algorithms. pp. 2896--2915. SIAM (2020)

\bibitem{ConfortiGP}
Conforti, M., Gerards, B., Pashkovich, K.: Stable sets and graphs with no even
  holes. Mathematical Programming  \textbf{153}(1),  13--39 (2015)

\bibitem{DF94}
Dyer, M., Frieze, A.: Random walks, totally unimodular matrices, and a
  randomised dual simplex algorithm. Mathematical Programming
  \textbf{64}(1-3),  1--16 (1994)

\bibitem{EV17}
Eisenbrand, F., Vempala, S.: Geometric random edge. Mathematical Programming
  \textbf{164}(1-2),  325--339 (2017)

\bibitem{GKS95}
Grossman, J.W., Kulkarni, D.M., Schochetman, I.E.: On the minors of an
  incidence matrix and its smith normal form. Linear Algebra and its
  Applications  \textbf{218},  213--224 (1995)

\bibitem{KO13}
Kawarabayashi, K.i., Ozeki, K.: A simpler proof for the two disjoint odd cycles
  theorem. J. Combin. Theory Ser. B  \textbf{103}(3),  313--319 (2013).
  \doi{10.1016/j.jctb.2012.11.004},
  \url{https://doi.org/10.1016/j.jctb.2012.11.004}

\bibitem{MoharThom}
Mohar, B., Thomassen, C.: Graphs on surfaces. Johns Hopkins University Press,
  Baltimore, U.S.A. (2001)

\bibitem{NT74}
Nemhauser, G.L., Trotter, J.L.E.: Properties of vertex packing and independence
  system polyhedra. Math. Programming  \textbf{6},  48–--61 (1974)

\bibitem{PSW}
Paat, J., Schl{\"o}ter, M., Weismantel, R.: Most {IP}s with bounded
  determinants can be solved in polynomial time.
  \url{http://arxiv.org/abs/1904.06874}  (2019)

\bibitem{seymour95}
Seymour, P.D.: Matroid minors. In: Handbook of combinatorics, {V}ol. 1, 2, pp.
  527--550. Elsevier Sci. B. V., Amsterdam (1995)

\bibitem{slilaty07}
Slilaty, D.: Projective-planar signed graphs and tangled signed graphs. J.
  Combin. Theory Ser. B  \textbf{97}(5),  693--717 (2007).
  \doi{10.1016/j.jctb.2006.10.002},
  \url{https://doi.org/10.1016/j.jctb.2006.10.002}

\bibitem{Tardos86}
Tardos, E.: A strongly polynomial algorithm to solve combinatorial linear
  programs. Operations Research  \textbf{34}(2),  250--256 (1986)

\bibitem{VC09}
Veselov, S.I., Chirkov, A.J.: Integer program with bimodular matrix. Discrete
  Optimization  \textbf{6}(2),  220--222 (2009)

\bibitem{WalterTruemper}
Walter, M., Truemper, K.: Implementation of a unimodularity test. Math.
  Program. Ser. C  \textbf{5}(1),  57--73 (2013).
  \doi{10.1007/s12532-012-0048-x},
  \url{http://dx.doi.org/10.1007/s12532-012-0048-x}

\end{thebibliography}
\bibliographystyle{splncs04}

\end{document}